\documentclass[submission,copyright,creativecommons]{eptcs}
\providecommand{\event}{NCL'22} % Name of the event you are submitting to
\usepackage{breakurl}             % Not needed if you use pdflatex only.
\usepackage{underscore}           % Only needed if you use pdflatex.

\usepackage{graphicx}
\usepackage{amsmath,amsthm,amssymb,proof,rotating,xcolor,mdframed}

\usepackage{hyperref}
\usepackage{tikz}
\usepgflibrary{arrows}
\usetikzlibrary{matrix,arrows,automata,decorations,fit,backgrounds}
\usetikzlibrary{decorations.pathreplacing,shapes.geometric,decorations.text}

\newcommand{\imp}{\vartriangleright}
\newcommand{\tto}{\blacktriangleright}

\newcommand{\R}{\mathcal{R}}

\newcommand{\W}{\mathcal{W}}
\newcommand{\N}{\mathcal{N}}
\newcommand{\MM}{\mathbf{M}}
\newcommand{\I}{\mathcal{I}}

\DeclareSymbolFont{symbolsC}{U}{txsyc}{m}{n} 
\DeclareMathSymbol{\sif}{\mathrel}{symbolsC}{74}

 \newtheorem{theorem}{Theorem}
\newtheorem{lemma}[theorem]{Lemma}

\theoremstyle{definition}

\theoremstyle{definition}
\newtheorem{definition}[theorem]{Definition}

\title{Non-Normal Super-Strict Implications}
\author{%Rob van Glabbeek
Guido Gherardi \qquad\qquad Eugenio Orlandelli
\institute{Department of Philosophy and Communication Studies, University of Bologna\\
Via Zamboni 38, Bologna, Italy}
\email{\quad guido.gherardi@unibo.it \quad\qquad eugenio.orlandelli@unibo.it}
}
\def\titlerunning{Non-normal SSI}
\def\authorrunning{G. Gherardi \& E. Orlandelli}
\begin{document}
\maketitle

\begin{abstract}
This paper   introduces the  logics of super-strict implications that are based on C.I. Lewis' non-normal modal logics $\mathsf{S2}$ and $\mathsf{S3}$.  The semantics of these logics is based on Kripke's semantics for non-normal modal logics. This solves a question we left open in a previous paper by showing that these logics are weakly connexive.\end{abstract}

\section{Introduction}
It is widely recognised,  that implication has a fundamental role in logic and deductive reasoning---see, e.g., \cite{Quine51,Rasiowa1974,Routley82}---and, hence, in philosophy `since the essence of philosophy, and especially of dialectics, lies in argumentation and reasoning' \cite[p. 1]{Routley82}. Nonetheless,  there is no general agreement about the correct formal rendering of implication. The  well-known rendering of implication as material implication ($\supset$) has been  defined, among others, by Russell \cite[p. 75--6]{PM}:

\begin{quote}
The essential property that we require of implication is this: `What is implied by a true proposition is true'. It is in virtue of this property that implication yields proofs. But this property by no means determines whether anything, and if so what, is implied by a false proposition. What it does determine is that, if $p$ implies $q$, then it cannot be the case that $p$ is true and $q$ is false, i.e. it must be the case that either $p$ is false or $q$ is true. The most convenient interpretation of implication is to say, conversely, that if either $p$ is false or $q$ is true , then `$p$ implies $q$' is to be true. Hence `$p$ implies $q$' is to be defined to mean: `Either $p$ is false or $q$ is true'. \end{quote}

\noindent Accordingly, in classical logic $A\supset B$ is  taken to be equivalent with $\neg A\lor B$: it is true whenever $A$ is false or $B$ is true. Even if someone  agrees with Russell on the essential property of implication, there is room to be dissatisfied with material implications having a false antecedent. Many found this explication of implication wanting because it falls prey to the so-called \emph{paradoxes of material implication}: 

\begin{equation}\tag{PMI}\label{PMI}
\neg A\supset( A\supset B)\hspace{3cm}B\supset(A\supset B)\end{equation}

\noindent C.I. Lewis \cite{Lewis1918,LL} argued that $\supset$ with its paradoxes  is acceptable for some uses of implication but not  for others, in particular for those where implication expresses a relation of logical entailment. Thus,  he proposed to  supplement classical logic with  a stronger implicative operator called \emph{strict implication} ($\sif$). The formula  $A\sif B$ is true if it is impossible that $A$ is true while $B$ is false, or equivalently if it is necessary that $A\supset B$. Semantically,  strict implication can be expressed in Kripke's relational semantics for modal logics by taking $A\sif B$ to be equivalent to $\Box(A\supset B)$: it is true at a point $w$ if $B$ is true at every point accessible from $w$ where $A$ is true.
 
 Some connexivists \cite{N30} and relevantists \cite{AB75,Routley82} argued that  we should not accept $\sif$ as a formal explication of entailment-expressing uses of implication because of  the so-called \emph{paradoxes of strict implication}: \begin{equation}\tag{PSI}\label{PSI}
\bot\sif B\hspace{3cm} A\sif\top
\end{equation}
These formulas are valid in Lewis' systems of strict implication. Moreover,  also the modalised version of these paradoxes---\mbox{$\Box \neg A\supset(A\sif B)$} and $\Box B\supset(A\sif B)$---are valid. Connexivists and relevantists quarrel with these formulas because, roughly, they take $A$ to entail---and thus imply---$B$ only if the truth of $B$ is  connected to/dependent on the truth of $A$.

One way to resist this line of thought is to claim that the entailment relation holds whenever an argument is such that it is impossible for its premisses to be true while its conclusion is false, cf. \cite[p. 203]{HC96}. From this perspective the paradoxes of strict implication should be accepted as sound principles and $\sif$ can  be taken as a correct explication of entailment. Lewis \cite[p. 250]{LL} proposed also a more  formal argument---his  so-called \emph{independent proof}---for the soundness of PSI. This proof is based on accepting the following  inferential principles:
\begin{description}
\item[$\wedge E$.] From a conjunction we can derive each one of its conjuncts.
\item[$\lor I$.] From a  formula $A$ we can derive each disjunction having $A$ as one of its disjuncts.
\item[\it{DisjSill}.] From a disjunction and the negation of one of its disjuncts we can derive the other disjunct.
%\item[\it{Trans}.] From the facts that $A$ implies $B$ and that $B$ implies $C$ we can derive that $A$ implies $C$.
\end{description}

\noindent These principles are all the ingredient needed for the following derivation of an arbitrary formula $B$ from a contradiction $A\wedge \neg A$ (i.e., they are sufficient to validate the \emph{ex falso quodlibet}):
$$
\infer[{\it DisjSill}]{B}{\infer[\lor I]{A\lor B}{\infer[\wedge E]{A}{A\wedge \neg A}}&\infer[\wedge E]{\neg A}{A\wedge\neg A}
}
$$
Moreover, if we take strict implication to be an object language representation of the underlying derivability relation---i.e., if we assume the deduction theorem for $\sif$---then we must accept  the following version of the first paradox of strict implication: $(A\wedge\neg A)\sif B$. Neither connexivists nor relevantists are convinced by these arguments. First, they do not accept that $A$ entails $B$ whenever it is impossible for $A$ to be true and $B$ false: this explication of entailment doesn't account for the fact that an entailment needs a connection between its antecedent(s) and its consequent. Second, they reject the \emph{ex falso quodlibet} and at least one of the principles used in Lewis' independent proof: connexivists usually reject $\wedge E$ and/or the law of \emph{simplification}: $(A\wedge B)\sif A$;\footnote{ Routley takes the failure of simplification as a distinctive feature of connexive logics, but this is not so since, e.g.,  Wansing's system $\mathsf{C}$ rejects {\em DisjSill} but accepts simplification and conjunctive elimination. The system $\mathsf{C}$  is interesting from the present perspective because its implication is both connexive and strict, cf. \cite{Wansing}. We are grateful to a reviewer for highlighting this connection.} relevantists usually reject  {\it DisjSill}.

The rejection of \emph{ex falso quodlibet}, be it based on the rejection of $\wedge E$ or of {\it DisjSill}, involves a major departure from classical logic since it means that at least one of $\wedge,\,\lor$, and $\neg$ does not satisfy its Boolean semantics. Thus neither connexivism nor relevantism can be seen as supplementation of classical logic. In \cite{GO} we proposed an alternative solution that sides with connexivists and relevantists in claiming that there are entailment-related uses of implication where PSI are not acceptable while at the same time siding with Lewis in expressing these uses of implication in a supplementation of classical logic. This proposal is based on  introducing  two strengthenings of normal strict implication called, respectively,  \emph{weak} ($\imp$) and \emph{strong} ($\tto$) \emph{super-strict implication} (SSI). The formula $A\imp B$ is defined to be true at a point  $w$ of a relational model if both $A\sif B$ and $\Diamond A$ are true at $w$.\footnote{If, moreover, $\Diamond\neg B$ is true at $w$ then $A\tto B$ is true therein. We will focus on $\imp$ since $\tto$ can be taken as a defined symbol.} Intuitively, $\imp$ has been introduced in \cite{GO} as expressing implications where the truth of the succedent depends, in some minimal sense, on that of the antecedent and it   accords with the following amendment of Russell's essential property of implication:
\begin{quote}
what is implied by a true and possible proposition is true. It is in virtue of this property that implication yields truths.
\end{quote}

Neither \ref{PSI} nor their modalised version are valid if we replace $\sif$ with $\imp$ given that the antecedent of these formulas is either   impossible (for the first paradox) or possibly contingent (for the second one). What becomes valid, instead, is the negation of the first paradox: $\neg(\bot\imp B)$. Moreover, as shown in \cite{GO}, all normal SSI are weakly connexive in that they validate \emph{Aristotle's Theses}:

\begin{equation}\tag{AT}\label{AT}
\neg(A\imp \neg A)\hspace{3cm} \neg (\neg A\imp  A)
\end{equation}

\noindent and  \emph{weak Boethius' Theses}:

\begin{equation}\tag{wBT}\label{BT}
(A\imp B)\supset\neg(A\imp\neg B)\qquad\ (A\imp \neg B)\supset\neg(A\imp B)
\end{equation}

\noindent SSI are a supplementation of classical logic just like  strict implication and modal logics. This means that the \emph{ex falso quodlibet} and all the principles used to derive it in Lewis' independent proof hold in logics of SSI. The difference with $\sif$ is that neither the deduction theorem nor the law of simplification holds for $\imp$. In these respects logics of SSI differ from most other connexive logics.\footnote{Another connexive system that rejects simplification while accepting $\wedge E$ is Nelson's \cite{N30}  one; nevertheless Nelson's connexive system has a conjunction that is weaker than the classical one, cf. \cite{MP19}.} Another striking difference is that, just like for Pizzi's \cite{PW97} \emph{consequential implication}, strong Boethius' thesis does not hold for $\imp$: the formula $ (A\imp B)\imp\neg(A\imp\neg B)$ is not valid in general and we don't know whether it is valid in some class of relational frames.\footnote{Claudio Pizzi suggested to call \emph{Boethian} logics validating \ref{AT} and \ref{BT} but not strong BT.} This is a price we paid to have a paradox-free and weakly connexive  implication that is  a simple modal supplementation of classical logic.

One question that was left open in \cite{GO} is whether non-normal logics of SSI are connexive and well-behaved like the normal ones. This question is interesting from both a logical and a philosophical perspective. Logically, it is interesting because, as far as we know, all modally definable connexive implications  are based on the relational semantics for  normal modalities even if the necessitation rule--- which allows to infer $\Box A$ from $A$---doesn't seem to have any relation with connexivity.\footnote{We won't consider here non-normal logics where at least one direction of $\Box (A\wedge B)\supset\subset(\Box A\wedge\Box B)$ fails.} Philosophically, it is interesting because $\imp$ can be seen as an amended version of Lewis' $\sif$, and it is well-known that---even if he introduced both   the normal systems of $\sif$ $\mathsf{S4}$ and $\mathsf{S5}$ and the non-normal ones $\mathsf{S1}$, $\mathsf{S2}$ and $\mathsf{S3}$---Lewis' favourite systems of strict implication were the non-normal systems $\mathsf{S2}$ and, to a lesser extent, $\mathsf{S1}$. System $\mathsf{S3}$ and its extensions validate the following form of transitivity $(A\sif B)\sif((B\sif C)\sif(A\sif C))$ that  Lewis \cite[p. 496]{LL} did not accept as a valid principle of implication. \emph{Mutatis mutandis}, cf. Sect. \ref{nnSSI}, this can be used as a motivation to choose $\mathsf{S2}$ as the best system of SSI.
This paper solves the question in \cite{GO} to the affirmative by introducing the $\mathsf{S2}$- and $\mathsf{S3}$-based logics of SSI and by showing that they are weakly connexive and well-behaved like the normal ones.\footnote{ Lewis' system $\mathsf{S1}$ will not be considered here since its treatment requires a much more complicated semantics that is based on a mixture of relational and neighbourhood semantics for non-normal logics, cf \cite{Cress72,TW}.}

The rest of the paper is organised as follows. Section \ref{nSSI} recaps the normal logics of SSI as they have been introduced in \cite{GO}. Section \ref{nnSI} sketches Lewis' non-normal logics of strict implication $\mathsf{S2}$ and $\mathsf{S3}$. Section \ref{nnSSI} introduces the $\mathsf{S2}$- and $\mathsf{S3}$-based non-normal logics of SSI and discusses some of their features.

%--------              NORMAL SSI
\section{Normal super-strict implication}\label{nSSI}
We sketch here the logics of normal  SSI.  For simplicity, we depart from the presentation given in \cite{GO} by treating $\tto$ as a defined symbol. The language $\mathcal{L}^\imp$ is defined by the following grammar, where $p\in\mathcal{P}$ for  a denumerable set of sentential variables $\mathcal{P}$,

\begin{equation}\tag{$\mathcal{L}^\imp$}\label{language}
A\;::=\; p\;|\;\bot\;|\;A\wedge A\;|\; A\lor A\;|\; A\supset A\;|\;A\imp A
\end{equation}

Parentheses follow the usual conventions and $\imp$ binds lighter than other operators. Capital roman letters will be used as metavariables for formulas. %The weight of a formula, $\texttt{w}(A)$, is given by the number of binary operators occurring therein. 
The symbol $\top$ is a short for $\bot\supset\bot$ and  $\neg A$ is short for $A\supset\bot$. Finally,
$\tto$ can be expressed in terms of $\imp$ according to the following definition:

\begin{equation}\tag{Def. $\tto$}\label{tto}
A\tto B\;\equiv\; ((A\imp B)\wedge (\neg B\imp\top))
\end{equation}

The semantics of SSI is the usual relational semantics for normal modalities, cf. \cite{Garson}. In particular a \emph{frame} is a pair $\mathcal{F}=\langle \mathcal{W},\mathcal{R}\rangle$, where $\mathcal{W}$ is a non-empty set of \emph{points} (or \emph{worlds}) and $\mathcal{R}\subseteq\mathcal{W}\times\mathcal{W}$ is binary \emph{accessibility relation} on $\mathcal{W}$. A model $\mathcal{M}$ is a frame augmented with an \emph{interpretation function} \mbox{$\mathcal{I}:\mathcal{P}\longrightarrow 2^\mathcal{W}$} mapping each sentential variables to the set of points where that variable is true. 

\emph{Truth of a formula} $A$ at a point $w$ of a model $\mathcal{M}$, to be written as $\models_w^\mathcal{M}A$, or simply as $\models_wA$ when the model is clear from the context, is defined as usual for sentential variables and for the extensional operators (e.g., a conjunction is true at a point $w$ whenever both conjuncts are true at $w$). Truth for formulas having $\imp$ as principal operator is defined as follows:\bigskip

\noindent\begin{tabular}{llll}\
%$\models_w^\mathcal{M} p$&iff&& $w\in V(p)$\\\noalign{\medskip}
%$\not\models_w^\mathcal{M} \bot$\\\noalign{\medskip}
%$\models_w^\mathcal{M} A\wedge B$&iff&&$\models_w^\mathcal{M} A$ and $\models_w^\mathcal{M} B$\\\noalign{\medskip}
%$\models_w^\mathcal{M} A\lor B$&iff&&$\models_w^\mathcal{M} A$ or $\models_w^\mathcal{M} B$\\\noalign{\medskip}
$\models_w^\mathcal{M} A\imp B$&iff& for all $v\in \mathcal{W}$, if $w \mathcal{R}v$  and $\models^\mathcal{M}_v A$, then $\models_v^\mathcal{M} B$&\&\\\noalign{\smallskip}
&&  some $v\in \mathcal{W}$ is such that $w\mathcal{R} v$ and $\models_v^\mathcal{M} A$
\end{tabular}\bigskip

\emph{Truth  in a model} ($\models^\mathcal{M} A$) is defined as truth on every point of that model; \emph{validity in a  frame} $\mathcal{F}$ ($\mathcal{F} \models A$) and \emph{validity in a class $\mathcal{C}$ of frames} ($\mathcal{C}\models A$) is defined as truth in every model based either on that  frame or in that class of frames, respectively; where a model is based on a frame if it is obtained  by adding an interpretation function to that frame. %Finally, the (local) \emph{consequence relation} with respect to a class of frames $\mathcal{C}$ is defined as follows (where $\Gamma$ is a set of formulas):

%\begin{quote}
%$\Gamma \models^\mathcal{C} A\quad$ iff \quad for each point $w$ of  a model based on a frame in $\mathcal{C}$, if all formulas in\ $\Gamma$ are \phantom{a}\hspace{2.4cm}true at $w$\ then  $A$ is true at $w$
%\end{quote}

An easy calculation shows that the truth-clause for $A\tto B$ would be as follows:\bigskip

\noindent\begin{tabular}{llll}\
%&&(iii)& there is some $v\in W$ such that $\not\models_v^\mathcal{M} B$
$\models_w^\mathcal{M} A\tto B$&iff& for all $v\in W$, if $w Rv$  and $\models^\mathcal{M}_v A$, then $\models_v^\mathcal{M} B$&\&\\\noalign{\smallskip}
&&  some $v\in W$ is such that $wR v$ and $\models_v^\mathcal{M} A$&\&\\\noalign{\smallskip}
&&  some $v\in W$ is such that $wR v$ and $\not\models_v^\mathcal{M} B$\end{tabular}\bigskip

\noindent Moreover, it is possible to express all of $\,\Diamond,\;\Box$, and $\sif$ in terms of $\imp$ according to these  definitions:

\begin{equation}\tag{Def. $\Diamond$}\label{Diamond}
\Diamond A\;\equiv A\imp \top\end{equation}
\begin{equation}\tag{Def. $\Box$}\label{box}
 \Box A\;\equiv\; \neg(\neg A\imp \top)
\end{equation}
\begin{equation}\tag{Def. $\sif$}\label{strictif}
A\sif B\;\equiv\; \neg(( A\wedge\neg B)\imp\top)
\end{equation}

\noindent Over models based on a serial frame  (i.e., such that $\forall w\in\mathcal{W}.\exists v\in\mathcal{W}. w\mathcal{R}v$) we can simplify the definition of $\Box A$ by expressing it as $\top\imp A$.

By a \emph{SSI-logic} we mean the set of all $\mathcal{L}^\imp$-formulas that are valid in some given class of frames. In \cite{GO} we studied the SSI-logics that are definable by some combination of the following properties of $\mathcal{R}$: seriality, reflexivity, transitivity, symmetry, and Euclideaness---i.e., for the modal logics in the cube of normal modalities \cite{Garson}.

A proof-theoretic characterization of SSI-logics has been given in \cite{GO} by means of labelled sequent calculi in the style of \cite{Negri05}. Alternatively, it is possible to use existing axiomatizations of normal modal logics, see \cite{Garson}, to give a proof-theoretic characterization of SSI-logics since $A\imp B$ be can be expressed in modal terms as $\Diamond A\wedge\Box(A\supset B)$. This fact has been exploited in \cite{GO} to show that the complexity of the satisfiability problem in a SSI-logic is  equivalent to that in the corresponding modal logic (where the modal logics $\mathsf{S5}$ and $\mathsf{KD45}$ are  {\sc NP}-complete and most of the other logics are {\sc Pspace}-complete).

We end this section by presenting some facts about SSI-logics \cite{GO}.

\begin{lemma}\label{lemmanormal}\
\begin{enumerate}
\item Neither the  paradoxes of material implication nor those of strict implication are valid for $\imp$.
\item The negation of the first paradox of strict implication is valid for $\imp$.
\item Aristotle's theses and weak Boethius' theses are valid for $\imp$.
\item Strong Boethius' theses are not valid for $\imp$.
\item Neither reflexivity $A\imp A$ nor contraposition $(A\imp B)\supset (\neg B\imp\neg A)$ hold.
\item Modus ponens holds for $\imp$ only over classes of frames that are reflexive: the following rule $$\infer{\mathcal{C}\models B}{\mathcal{C}\models A\imp B\qquad& \mathcal{C}\models A}$$ preserves validity only if $\mathcal{C}$ is a class of reflexive frames.
\end{enumerate}
\end{lemma}

%----------		NON-NORMAL STRICT IMPLICATIONS
\section{Non-normal strict implication}\label{nnSI}

The language  of strict implication $(\mathcal{L}^\sif$) is determined by the following grammar (where $p\in\mathcal{P}$):
\begin{equation}\tag{$\mathcal{L}^\sif$}\label{languagestrict}
A\;\;::=\;\; p\;|\;\bot\;|\;A\wedge A\;|\; A\lor A\;|\; A\supset A\;|\;A\sif A
\end{equation}
We use the same conventions we introduced for $\mathcal{L}^\imp$; and  we define the normal modalities  as follows:

\begin{equation}\tag{Def.$^\sif$ $\Box$}\label{Boxstrict}
 \Box A\;\equiv\; \top\sif A
\end{equation}
\begin{equation}\tag{Def.$^\sif$ $\Diamond$}\label{Diamondstrict}
\Diamond A\;\equiv \neg(\top\sif \neg A)\end{equation}
In talking about Lewis' systems we will often move from the language \ref{languagestrict} to the standard modal language $\mathcal{L}^\Box$ having $\Box$ and $\Diamond$ as primitive operators and $\sif$ as an abbreviation for $\Box(A\supset B)$.

Lewis' systems $\mathsf{S2}$ and $\mathsf{S3}$ are not normal systems of modal logic in that they do not validate the unrestricted necessitation rule

\begin{equation}\tag{$RN$}\label{RN}
\infer{\Box A}{A}
\end{equation}
\noindent but only a restricted version thereof, called $RN^{rest}$, where $A$ has to be  a theorem of the classical propositional calculus $\mathsf{PC}$. Consequently, these systems contain no theorem of shape $\Box\Box A$ and we must use a semantics that is a generalisation of the relational semantics for normal modalities.
A simple relational  semantics for Lewis' non-normal systems $\mathsf{S2}$ and $\mathsf{S3}$ has been introduced in \cite{kripke}. It is based on the fact that these systems can be consistently extended with the formula $\Diamond\Diamond A$:  we can have $\mathsf{S2}$- and $\mathsf{S3}$-models with points where every formula, $\bot$ included, is possible. Kripke used this fact to extend the well-known relational semantics for normal modalities with  so-called \emph{non-normal points} where every formula is possible and no formula is necessary, see \cite{CZ97,HC96,kripke,Tesi21}.

Formally, this means that a non-normal frame is a triple $\mathbf{F}=\langle \mathcal{W},\mathcal{R},\mathcal{N}\rangle
$
 where $\mathcal{W}$ and $\mathcal{R}$ are as for normal modalities, see \S \ref{nSSI}, and $\mathcal{N}$ is a subset of $\mathcal{W}$ (containing the \emph{normal points} of the model).\footnote{As in \cite{CZ97,Tesi21},  we lift the restrictions given in \cite{HC96,kripke} that (i) a frame must contain some normal point and that (ii) $\forall w\in\mathcal{W}/\mathcal{N}.\exists u\in\mathcal{N}.u\mathcal{R}w$.} A non-normal model $\mathbf{M}$ is a frame together with an interpretation function $\mathcal{I}$ mapping propositional variables to subsets of $\mathcal{W}$. Truth of a formula at a point of a model is defined as usual for extensional formulas and it is defined as follows for $\sif$:\bigskip
 
 \noindent\begin{tabular}{llll}\
$\models_w^\mathbf{M} A\sif B$&iff& $w\in\mathcal{N}$ and for all $v\in \mathcal{W}$, if $w \mathcal{R}v$  and $\models^\mathbf{M}_v A$, then $\models_v^\mathbf{M} B$\end{tabular}\bigskip

Given  Definitions \ref{Boxstrict} and \ref{Diamondstrict}, we have the following truth clauses for $\Box$ and $\Diamond$:\bigskip

\noindent\begin{tabular}{llll}
$\models_w^\mathbf{M} \Box A$&iff& $w\in\mathcal{N}$ and  for all $v\in \mathcal{W}$, if $w \mathcal{R}v$   then $\models_v^\mathbf{M} A$\\\noalign{\medskip}
$\models_w^\mathbf{M} \Diamond A$&iff&  $w\not\in\mathcal{N}$ or  some  $v\in \mathcal{W}$ is such that  $w \mathcal{R}v$  and $\models^\mathbf{M}_v A$\end{tabular}\bigskip

The definition of truth in a model, and consequently that of validity in a (class of) frame(s), differs from the usual one, see \S \ref{nSSI}, in that we restrict quantification to normal points: a formula is true in $\mathbf{M}$ if it is true in all normal points of the model (so that whether $A$ is true  in a non-normal point  is irrelevant for evaluating its truth in a model and its validity).

The logic $\mathsf{S2}$ is defined as the set of all formulas that are valid in reflexive non-normal frames and $\mathsf{S3}$ as the set of formulas valid in reflexive and transitive non-normal frames. The logic defined by the class of all non-normal frames is called $\mathsf{S2^0}$. We refer the reader to \cite{HC96} for other logics defined in terms of classes of non-normal frames.

As we have already anticipated the rule \ref{RN} does not preserve validity in these classes of frames. To illustrate, the formula $\Box\top$ is valid in $\mathsf{S2}^0$-frames but $\Box\Box\top$ is not valid:  a frame $\mathbf{F}$ containing both a normal point $w$ and a non-normal point  that is accessible from $w$ is such that $\mathbf{F}\models \Box \top$ and $\mathbf{F}\not\models\Box\Box\top$.

A proof-theoretic characterization of $\mathsf{S2}$ and $\mathsf{S3}$ in terms of labelled sequent calculi defined over the language $\mathcal{L}^\Box$ has been given in \cite{Tesi21}. Lewis' \cite{LL} axiomatization of $\mathsf{S2}$ is as follows (taking $\lor$, $\supset$, and $\sif$ as defined symbols and using $A[B/p]$  ($A[B//p]$) for the formula obtained by replacing, every (some, resp.) occurrence of $p$ in $A$ with an occurrence of $B$).\medskip

\begin{definition}[Lewis' axiomatisation of $\mathsf{S2}$]\label{Lewisax}\

\begin{minipage}[t]{0,5\textwidth}\begin{itemize}
\item Axioms:
\begin{enumerate}
\item $(p\wedge q)\sif (q\wedge p)$
\item $(p\wedge q)\sif p$
\item $p\sif (p\wedge p)$
\item $((p\wedge q)\wedge r)\sif(p\wedge(q\wedge r))$
\item$((p\sif q)\wedge(q\sif r))\sif (p\sif r)$
\item $(p\wedge(p\sif q))\sif q$
\item $\Diamond (p\wedge q)\sif \Diamond p$
\end{enumerate}
\end{itemize}\end{minipage}
\begin{minipage}[t]{0,5\textwidth}\begin{itemize}
\item Rules:
\begin{enumerate}
\item Uniform substitution:\quad$\infer{A[B/p]}{A}$
\item Substitution of strict equivalents:$$\infer{A[B//C]}{A\qquad (B\sif C)\wedge (C\sif B)}$$
\item Adjunction: \quad $\infer{A\wedge B}{A\qquad B}$
\item Strict detachment:\quad $\infer{B}{A\sif B\qquad A}$
\end{enumerate}

\end{itemize}
\end{minipage}\end{definition}\bigskip

\noindent The calculus  $\mathsf{S3}$ is obtained by replacing  axiom 7 with: $(p\sif q)\sif(\neg\Diamond q\sif\neg\Diamond p)$.\footnote{Axiom 7 is a theorem of $\mathsf{S3}$; if we remove axiom 7 we obtain Lewis' system $\mathsf{S1}$.} One peculiarity of Lewis' axiomatic systems is that they are not given as extensions of the classical propositional calculus $\mathsf{PC}$ (even if they contain it). A more customary axiomatization of $\mathsf{S2}$, based on $\mathsf{PC}$ and taking $\Box$ as primitive modal operator, has been given by Lemmon \cite{lemmon}.

\begin{definition}[Lemmon's axiomatisation of $\mathsf{S2}$]\label{Lemmonax}\

\begin{minipage}[t]{0,5\textwidth}\begin{itemize}
\item Axioms:
\begin{enumerate}
\item[$PC.$] All theorems of $\mathsf{PC}$
\item[$K.$] $\Box (A\supset B)\supset(\Box A\supset\Box B)$
\item[$T.$] $\Box A\supset A$
\end{enumerate}
\end{itemize}\end{minipage}
\begin{minipage}[t]{0,5\textwidth}\begin{itemize}
\item Rules:
\begin{enumerate}
\item[$MP.$] Modus Ponends:\quad$\infer{B}{A\supset B\qquad A}$
\item[$BR.$] Becker's rule:\quad$\infer{\Box(\Box A\supset \Box B)}{\Box(A\supset B)}$
\item[$N^{rest}.$] Restricted necessitation: \quad $\infer{\Box A}{A}$\\ provided $A$ is a theorem of $\mathsf{PC}$ 
\end{enumerate}

\end{itemize}
\end{minipage}\end{definition}\bigskip

\noindent
Lemmon's axiomatization of $\mathsf{S3}$ is obtained by replacing axiom $K$ with $\Box(A\supset B)\supset\Box(\Box A\supset\Box B)$ and by dropping Becker's rule.\footnote{Axiom $K$ is derivable and Becker's rule is admissible in $\mathsf{S3}$.} An axiomatization of the logic $\mathsf{S2^0}$ is obtained by dropping the reflexivity axiom $T$ from Lemmon's axiomatization of $\mathsf{S2}$.

%---------			NON-NORMAL SSI
\section{Non-normal super-strict implication}\label{nnSSI}
We are now in a position to introduce non-normal logics of SSI. We will use the language \ref{language} and, as will be clear from the semantics, we can still define all of $\tto$, $\Box$, $\Diamond$, and $\sif$ in terms of $\imp$ as we did in \S \ref{nSSI}.\footnote{ And, \emph{vice versa}, $\imp$ is definable in terms of any one of $\Box$, $\Diamond$ and $\sif$, but not in terms of $\tto$.} The \emph{weight} of a formula,  $\texttt{w}(A)$, is defined as the number of binary operators occurring therein; the \emph{modal depth} of a formula, $\texttt{md}(A)$, is defined as the maximal nesting of modalitites---$\imp,\,\tto,\,\Box,\,\Diamond,$ and $\sif$---in $A$.

The semantics for non-normal SSI is based on Kripke's one for the logics $\mathsf{S2}$ and $\mathsf{S3}$. In particular,
a \emph{frame} is a triple $\mathbf{F}=\langle\W,\R,\N\rangle$, where $\W\neq\emptyset$,\; $\R\subseteq\W\times\W$, and $\N\subseteq \W$;  a \emph{model} $\MM$ is a frame augmented with an interpretation function $\I:\mathcal{P}\longrightarrow\W$. The truth clause for $A\imp B$ is as follows:\bigskip

\noindent\begin{tabular}{llll}\
%&&(iii)& there is some $v\in W$ such that $\not\models_v^\mathcal{M} B$
$\models_w^\mathbf{M} A\imp B$&iff& for all $v\in W$, if $w Rv$  and $\models^\mathbf{M}_v A$, then $\models_v^\mathbf{M} B$&\&\\\noalign{\smallskip}
&&  some $v\in W$ is such that $wR v$ and $\models_v^\mathbf{M} A$&\&\\\noalign{\smallskip}
&&  $w\in\N$\end{tabular}\bigskip

\noindent \emph{Truth in a model} is defined as truth in every normal point of that model and \emph{validity} in a (class of) frame(s) is defined as truth in every model based on that (class of) frame(s). 

As for Lewis' systems we will consider the non-normal $\mathcal{L}^\imp$-logics $\mathsf{S2^0},\,\mathsf{S2},$ and $\mathsf{S3}$, which are defined as the set of formulas valid in the class of all frames, all reflexive frames and  all preordered frames, respectively.

As we have already anticipated, in non-normal logics $\imp$ is interdefinable with any one of $\sif$, $\Box$ and $\Diamond$ just like in normal logics, cf. Section \ref{nSSI}. As a consequence we have that Lewis' and Lemmon's axiomatisations of $\mathsf{S2^0}$, $\mathsf{S2}$ and $\mathsf{S3}$ are axiomatisations of the corresponding non-normal logics of SSI (over the $\mathcal{L}^\sif$- and the $\mathcal{L}^\Box$-language, respectively) and that these logics of SSI are decidable, cf. \cite{decid} for the decidability of $\mathsf{S2}$ and $\mathsf{S3}$. 

We move on to present some interesting features of logics on non-normal SSI which, all in all, will show that these logics have a behaviour that resembles that of normal ones. We begin by listing some important formulas that do not belong to these logics.

\begin{lemma}\label{unvalid} None of the following formulas is valid in classes of non-normal frames:
\begin{enumerate}
\item The paradoxes of material implication: $\neg A\imp(A\imp B)$ and $B\imp(A\imp B)$;
\item The paradoxes of strict implication: $\bot\imp B$ and $A\imp \top$;
%\item Strong Boethius' thesis: $ (A\imp B)\imp\neg (A\imp\neg B)$;
\item Reflexivity: $A\imp A$;
\item Contraposition: $(A\imp B)\supset (\neg B\imp\neg A)$.
\end{enumerate}
\end{lemma}
\begin{proof} These formulas are not valid in normal frames, cf. Lemma \ref{lemmanormal}, and, hence, cannot be valid in frames generalising them such as non-normal frames.
\end{proof}

Next we show that these logics are weakly connexive in that they validate Aristotle's and weak Boethius' theses, though in general they do not validate strong Boethius' thesis. 
\begin{theorem}\label{connexive}
Aristotle's  and weak Boethius' theses belong to the $\mathcal{L}^\imp$-logic $\mathsf{S2^0}$.
\end{theorem}
\begin{proof} The lemma follows from the fact that \ref{AT}  and \ref{BT} are valid in all normal frames, cf. Lemma \ref{lemmanormal}. These theses have modal depth 1  and  for formulas of modal depth 1 validity in non-normal frames coincide with validity in normal ones (since we are restricting attention to normal points). 
\end{proof}

 Next, we consider the status of Lewis' axioms, see Definition \ref{Lewisax}, when we replace $\sif$ with $\imp$. This is important to better understand how $\imp$ differs from $\sif$.
\begin{lemma}\label{Lewisaxsuper} None of the axioms of Lewis' axiomatization of $\mathsf{S2}$ (nor of $\mathsf{S3}$) is valid if we replace (all) instances of $\sif$ with instances of $\imp$.\footnote{ Nonetheless, $\imp$ is transitive in that the formula $((A\imp B)\wedge (B\imp C))\supset(A\imp C)$ is valid in every  non-normal frame.}
\end{lemma}
\begin{proof}
It is immediate to notice that the antecedent of each one of Lewis' axioms is such that we can easily define a point of a model---be it reflexive or not---where one of its instances is not possible, and this is enough to conclude that these formulas are not valid if we replace $\sif$ with $\imp$.
\end{proof}

\noindent Things goes differently with Lewis' rules since all of them preserve validity over reflexive frames, where reflexivity is needed only for super-strict detachment (just as for strict detachment). 

The results in Lemma \ref{Lewisaxsuper} might appear rather disappointing and at first glance they can be taken as a reason to dismiss SSI as an interesting explication of entailment-related implications.  Even if it might be acceptable that simplification (Axiom 2) does not hold, it would be difficult to argue that entailment is not closed under commutativity and associativity of $\wedge$ (Axioms 1 and 4) and  duplication (Axiom 3); this would be  even harder for transitivity of implication (Axiom 5), the formula version of  detachment (Axiom 6) and the consistency axiom (Axiom 7).  It seems that we amended $\sif$ in such a way that the price paid to avoid the paradoxes of strict implication is too high. Nevetheless, we believe this conclusion is too fast. In general, we should separate formulas into impossible, necessary and  contingent ones. Anyone believing that SIP are not valid principles of entailment should agree that Lewis' axioms should not be accepted when they have an impossible antecedent---e.g.,  a contradictory conjunction does not entail anything irrespectively of the order of its conjuncts and, hence, entailment is not closed under the commutativity of impossible conjunctions.  When, instead, we consider instances of Lewis' axioms with a necessary antecedent then we should accept them as sound principles of entailment. These instances pose  no threat since they are valid when we replace $\sif$ with $\imp$ whenever we consider reflexive frames---every necessary formula is also possible over reflexive frames---and there are independent reason to restrict our attention to reflexive frames---e.g., super-strict detachment preserves validity only over reflexive frames and we took the truth of $A$ and of $A$ implies $B$ to  yield the truth of $B$. The problem lies in instances having a contingent antecedent. These instances are sound principles of entailment but they are not   valid for $\imp$: they are true at worlds where the antecedent is possible and false at worlds where it is not possible. Even if the antecedent in general is contingent, in worlds where it is not possible it expresses a relative impossibility and, hence, the entailment relation should not hold just like instances with an (absolutely) impossible antecedent. Things are different in worlds where the contingent antecedent is possible and its truth  necessitates that of the succedent. This is a paradigmatic situation where the antecedent should imply the consequent according to the entailment-related uses of implication we aim at capturing. In this situation Lewis' axioms are true but not valid. Nevertheless we can transform each one of Lewis' axioms into a valid formula governing $\imp$ if we restrict our attention to instances having a possible antecedent---modulo reflexivity for Axiom 6.

\begin{lemma}\label{Lewisaxsupergood} If $A\imp B$ is one of Lewis' axioms given in Definition \ref{Lewisax} (with $\sif$ replaced by $\imp$) then the formula $\Diamond A\supset(A\imp B)$ is valid over reflexive non-normal frames.
\end{lemma}
\begin{proof}
The definition of $\imp$ coincides with that of $\sif$ whenever we have an implication whose antecedent is possible.\end{proof}

\noindent This lemma shows that if we replace $\sif$ with $\imp$ because we don't accept paradoxical instances of PSI then we obtain as valid exactly all and only the non-paradoxical instances of Lewis' implications.\footnote{Observe that $\Diamond A\supset(A\imp B)$ can be used as an alternative definition of strict implication in the language $\mathcal{L}^\imp$. Nevertheless, we take the main role of this formula role  to be that of separating paradoxical and non paradoxical instances of  valid strict  implications; where, according to the present perspective, only the latter can be taken to express an entailment related implication.} 
 It also shows that the correct logic of SSI should be at least as strong as $\mathsf{S2}$ since we must restrict our attention to reflexive frames if we want super-strict detachment---as well as its formula-version--- to hold. Moreover, we can use an analogous argument to show that the correct logic of SSI shouldn't be stronger than $\mathsf{S2}$. Indeed, in $\mathsf{S3}$ the formula $\Diamond(A\imp B)\supset((A\imp B)\imp(\Diamond (B\imp C)\supset((B\imp C)\imp(\Diamond A\supset(A\imp C)))))$ is valid, but this formula is not better motivated than $(A\sif B)\sif((B\sif C)\sif(A\sif C))$ in Lewis' approach.

We conclude this section by noticing that also any dissatisfaction with the non-reflexivity of $\imp$ can be mitigated along the same line since $\Diamond A\supset(A\imp A)$ is valid. What is not valid is $\neg \Diamond A\supset(A\imp A)$, but instances of reflexivity of $\imp$  with a (locally) impossible antecedent should not be accepted as valid according to the present paradox-free explication of implication and entailment.

\section{Conclusion}

This paper presented the non-normal logics of SSI $\mathsf{S2}$ and $\mathsf{S3}$ as a way to amend Lewis' $\sif$ and  express  paradox-free notions of implication and entailment. The semantics for these logics is based on Kripke's semantics for non-normal logics and, this paper has shown, non-normal logics of SSI are as well-behaved as normal ones. In particular they validate the same connexive principles and can thus be seen as a simple and natural family of non-normal connexive supplementations of classical logics. 

This paper has not addressed the problem of giving a proof-theoretic characterisation of these logics, if not by translating them into the language $\mathcal{L}^\sif$ or $\mathcal{L}^\Box$.  If we want to stay within the language $\mathcal{L}^\imp$ then labelled calculi for these logics can be given by a mixture of the calculi in \cite{GO} with those in \cite{Tesi21}. A characterisation in terms of axiomatic systems can be given by exploiting the techniques used in \cite{R20}. We believe this latter problem to be of particular interest given the results of Lemma \ref{Lewisaxsuper}, but we leave it for future research.

\subsection*{Acknowledgements}
We are grateful to the anonymous reviewrs and to audience at \emph{Trends in Logic XXI} for valuable feedback.
%-----------------  		BIBLIOGRAPHY
\bibliographystyle{eptcs}
\bibliography{biblio}
\end{document}